\documentclass[a4paper]{article}
\usepackage{fullpage,url,amsmath,amsfonts,amssymb, tedmath}
\usepackage[stable]{footmisc}
\usepackage{enumitem}
\usepackage{fullpage,url,amsmath,amsfonts,amssymb, tedmath}
\usepackage{enumitem}
\usepackage{amsmath,amsfonts,amssymb, tedmath}
\newcommand{\all}{\beta}

\newcommand{\cc}{\mathbb{C}}
\title{Coding theory:  
the  unit-derived methodology.}
\author{Ted
 Hurley\footnote{National Universiy of Ireland Galway, email:
 Ted.Hurley@NuiGalway.ie} 
   \quad and Donny Hurley\footnote{Institute of Technology Sligo, email: hurleyd@yahoo.com}}
\date{} 
\begin{document}
\maketitle

\begin{abstract}\let\thefootnote\relax
 
The unit-derived method in coding
 theory is  shown to be 
 a unique optimal scheme for constructing
 and analysing codes.  In many cases efficient and practical decoding
 methods are produced. 
Codes  with efficient 
  decoding algorithms at maximal distances possible are  
 derived from unit schemes.  
 In particular unit-derived codes from Vandermonde or Fourier 
 matrices are particularly commendable giving rise to mds codes of
 varying rates with   
 practical and efficient decoding algorithms.

For a given rate and given error correction capability, explicit 
 codes with efficient error correcting algorithms are designed to these
 specifications. An explicit  constructive proof with  an efficient
 decoding algorithm is given for Shannon's theorem. For a given finite
 field, codes are constructed which are `optimal' for this field.  

\end{abstract}

\section{Introduction and background}
Error-correcting codes 
are used extensively in communications' applications 
including digital video,
radio, mobile communication, satellite/space communications and other systems. 

Here the unit-derived method is exploited to design maximum distance
separable codes with efficient decoding algorithms.  For a given rate
and a given error-correcting capability,  codes  with efficient decoding
algorithms are designed  
to these specifications and are shown algebraically to have the required
properties. This is used to give explicit codes with efficient decoding algorithms  to  prove  Shannon's theorem.

Section \ref{layout} gives further details on content and results.  
Samples demonstrating the extent of the constructions are given.
Some well-known codes in practical use are shown to be special cases;
better performing ones  can be designed from the general techniques.

  Background on coding theory   
  may be found in \cite{blahut},\cite{mceliece} and others. Most of the
  algebraic background may  be found in \cite{blahut} and further
  background on algebra and coding theory is developed or referenced as
  required.

Now  $(n,r,d)$ denotes a code of length $n$, dimension $r$ and (minimum) distance $d$. The rate of the code is $\frac{r}{n}$. The code $(n,r,d)$ can correct $t=\floor{\frac{d-1}{2}}$ errors and this is the error-correction 
capability of the code. 
The code is called a {\em maximum distance separable} (mds) code if it of the form $(n,r,n-r+1)$, that is, if  it attains   the maximum distance allowable  for a given length and dimension. 

$GF(q)$ denotes the finite field of $q$ elements where $q=p^s$ is a
power of a prime $p$. The units of $GF(q)$ are the non-zero elements of
$GF(q)$ and these units form a cyclic group generated by a primitive
$(q-1)^{th}$ root of unity in $GF(q)$. For a prime $p$, $GF(p) =\Z_p$,
the integers modulo $p$.

\subsection{Unit-derived codes} In  \cite{hur1}, and  also in \cite{hur0,hur2},  methods are developed  for constructing 
{\em unit-derived  codes};  these  methods are fundamental.  The
unit-derived schemes may be described 
briefly  as follows. Let $R_{n\ti n}$ denote the ring of
 $n\ti n$ matrices with entries from $R$, a ring with identity,   often  a field but not restricted to such.
Suppose  $UV=I_{n\ti n}$ in $R_{n\ti n}$.   
Taking any $r$ rows of $U$
as a generator matrix defines an $(n,r)$ code  and a check
matrix is obtained by deleting the corresponding columns of $V$.
 Further details may be  found in expanded book chapter form in \cite{hur2}.

Now $R$ can be any ring with identity and it has been  useful to consider cases
other than fields; cases where $R$ is taken as a polynomial ring,  a
group ring or as a matrix ring has been useful in constructing different
types of codes such as LDPC codes or Convolutional codes, \cite{hur3},
\cite{hurley33}, \cite{hurconv5,hurconv2}. 


From the unit scheme $UV=I$,  the first $r$ rows in particular of $U$
may be taken  as the generator matrix of a code and then the last $(n-r)$
columns of $V$ give a check matrix for this code.
Thus if $UV= I_n$ and $U= \begin{pmatrix} A \\ B \end{pmatrix}$ for an $r\ti n$ matrix $A$ and an $(n-r) \ti n $ matrix $B$ and $V=(C,D)$ for an $n\ti r$ matrix $C$ and an $n\ti (n-r)$ matrix $D$, this gives 
$UV=\begin{pmatrix}A \\ B \end{pmatrix} (C,D) = I_n$  from which  
$\begin{pmatrix} AC & AD \\ BC & BD \end{pmatrix} =I_n$. 

Thus $AD = 0_{r \ti (n-r)}$ and $D\T$ is a check matrix for the $(n,r)$ code with generator matrix $A$. Note also that $AC=I_{r\ti r}$, the identity $r\ti r$ matrix, and this  will be useful later.

Any linear code is equivalent to a unit-derived code but there may not
be any advantage in using the equivalence. 

Using the unit-derived method has many advantages. Unit-derived codes are in general not ideals; 
cyclic and some other such codes are ideals in group rings.   
  Many different codes of various rates and with predetermined properties may be constructed from a single unit scheme. 
Properties of the units may be used to derive codes of particular types and/or with particular properties. 
From the  set-up,   more information on the code $\mathcal{C}$ is available 
than just its generator and check matrix. Here also efficient
decoding methods for certain unit-derived codes are established. 

In the unit scheme as above,   
$\begin{pmatrix} A \\ B \end{pmatrix} (C, D) = \begin{pmatrix} AC & AD \\ BC & BD \end{pmatrix}$, $A$ is taken as the generator of a code. If $\al A$ is a codeword then $\al A * C = \al$. Hence the original transmitted vector 
$\al$ may be   obtained by multiplying on the right by $C$ once the  errors have been eliminated by an error-correcting method.

\subsection{Layout and summary}\label{layout} 
General theorems, Theorems \ref{genthmvan}, \ref{genthmfour}, required for the constructions and decoding methods are stated in Section \ref{statements}; these are proved later in Section \ref{general}.

Section   \ref{prelim}  presents examples as an
introduction to, and illustration of, the general techniques resulting from Theorems \ref{genthmvan} and \ref{genthmfour}. 
These examples have interest in themselves,  have full distances and implementable practical decoding algorithms. The examples are far from exhaustive and could be considered as prototypes for many others. 

An illustrative example in Subsection \ref{decode} demonstrates the decoding  method which is later derived in general in Section \ref{decoding}.   

Section \ref{general} introduces  the general method and derives
background 
results from which the properties of the unit-derived codes may be
deduced and from which the decoding algorithms are created.  Results
on Vandermonde/Fourier 
matrices are developed;  
 unit-derived codes from these are
particularly commendable with schemes for deriving maximum distance separable  codes with practical decoding algorithms. Section \ref{decoding} derives the general decoding algorithms. 
 
 
Section \ref{rate} describes  the
 general method of constructing codes with required rate and required
 error-correcting capability; Section \ref{example1}, 
 gives examples of such  required yield constructions. 
Section \ref{Shannon} uses the methods to derive an explicit proof of
 Shannon's theorem with  an efficient decoding algorithm.

Section \ref{best} notes `optimal' codes
 for  a particular  finite field.

The use of the unit-derived method 
for defining and analysing particular types of codes such as LDPC (Low
density parity check) codes, Convolutional Codes and others is discussed
in Section \ref{recent}. Section \ref{mceliece} suggests using the codes
for cryptographic schemes. 



\section{Construction of  special types\label{recent}\footnote{This section is independent of the succeeding  sections.}}

Low density parity check (LDPC) codes and convolutional codes attract much attention. Unit schemes are and have been used to generate such codes by relating the prescribed properties to properties of the units from which they are derived.  

\subsection{Low density} A low density parity check (LDPC) code is a
linear code where the check matrix has {\em low density} which means
that each row and column has only a small number of non-zero entries
compared to the size of the matrix. 

An LDPC code may  be obtained from a unit scheme 
$UV  = I_n$.  
To do this,  we must be able to choose  columns of $V$ to form a (check) matrix
which has  low density compared to its  size. The columns of $V$ chosen
decide the rows of $U$ to be used in generating the code. See \cite{hur3}
and \cite{hurley33} for further details.   

One way to ensure that any choice of rows 
 will be an LDPC code is to ensure that
$V$ itself has low density in all its rows and columns. 
  Indeed from such a unit system with
$V$ of low density many (different) LDPC codes can be generated. 
It is also possible to find in general such  $V$ of low density 
so that the resulting
LDPC codes have no short cycles \cite{hurley33}; LDPC codes with no short
cycles in the check matrix  are known to perform well.

 It may be shown that an  LDPC code is  equivalent to one derived from a unit scheme. 

This method has been used successfully in \cite{hurley33} to generate large length  LDPC codes with excellent performances.

\subsection{Convolutional codes} The unit-derived method may be used
to describe, define and study properties of Convolutional Codes, see
\cite{hur3}, \cite{hurconv5} 
; here the unit schemes
are over certain rings other than fields, such as polynomial rings or
group rings. The reference \cite{hur3} in book chapter form is particularly written as an introduction to these methods.  

The constructions in \cite{hurconv2} may be considered as  unit-derived  convolutional code construction schemes which have parallels to the (linear)  block code unit-derived schemes developed here.    
\subsection{Using group rings} Using the embedding of a group rings into a group of matrices, \cite{hur7}, allows the construction of self-dual, dual-containing,  quantum codes, \cite{hurley44}, and other types from units in group rings. Cyclic codes are ideals in the group ring of the cyclic code. Unit-derived codes in general are not ideals.  
\subsection{McEliece type encryption}\label{mceliece} 
 The  codes that are or can be constructed from the unit-derived codes
 developed here can have large length, have good error capability and
 good decoding capability and are thus suitable candidates for McEliece
 type encryption \cite{mceliececrypt}. The problem with low rate data
 can be eliminated. Permutation of the rows and different selections may be used. This should be compared with the cryptographic schemes of \cite{hur100}. 

\section{Main general results}\label{statements} Statements of the 
 results from which the general constructions and decoding methods are
 derived  are given in this section. The proofs of these follow from work  in
 Sections \ref{general} and  \ref{decoding}. 

Recall that an mds, maximum distance separable, code is one of the form
$(n,r,n-r+1)$ in which the maximum possible distance is obtained for a given length and dimension,  see \cite{blahut} for details. 

\begin{theorem}\label{genthmvan} Let $V=V(x_1,x_2, \ldots, x_n)$ be a Vandermonde $n\ti n$ matrix over a field $\F$ with distinct and non-zero $x_i$. Let $\mathcal{C}$ be the unit-derived code obtained by choosing in order $r$ rows of $V$ in arithmetic sequence with difference $k$. If $(x_ix_j^{-1})$ is not a $k^{th}$ root of unity for $i\neq j$ then $\mathcal{C}$ is an $(n,r,n-r+1)$ mds code over $\F$. 

In particular the result holds for consecutive rows as then  $k=1$ and 
 $x_i\neq x_j$ for $i \neq j$.  
\end{theorem}

For Fourier matrices the following theorem is obtained:

\begin{theorem}\label{genthmfour} 

(i) Let $F_n$ be a Fourier $n\ti n$ matrix over a field $\F$. 
 Let $\mathcal{C}$ be the
 unit-derived code obtained by choosing in order $r$ rows of $V$ in
 arithmetic sequence with arithmetic 
difference $k$ and  $\gcd(n,k) = 1$. Then
 $\mathcal{C}$ is an mds $(n,r,n-r+1)$. In particular this is true when
 $k=1$ that is, when the $r$ rows are chosen in succession. 

(ii) Let $\mathcal{C}$ be as in part (i). Then there exist efficient encoding and decoding algorithms for $\mathcal{C}$.

\end{theorem}  

The decoding methods are based on the decoding methods used in
\cite{hurleycomp} in connection compressed sensing by solving underdetermined
systems using error-correcting codes. These decoding methods  themselves
are based on the error-correcting methods due to  Pellikaan \cite{pell}
which is a method of finding error-correcting pairs.  
  
The complexity of encoding and decoding can be  $\max \{ O(n\log
n), O(t^2)\}$ where $t= \floor{\frac{n-r}{2}}$, that is where 
$t$ is the error-correcting capability of the code. The complexity is
discussed in Section \ref{complexity}.

\section{Initial cases}\label{prelim} Initial cases are presented  as an
introduction to, and illustration of, the general
techniques.   


The examples have interest in themselves and have  practical decoding
algorithms.  They also serve as prototypes as to  how general 
and longer length mds codes with efficient decoding algorithms may
 be constructed  using the unit-derived method with  Vandermonde/Fourier
 matrices. For the proofs that the codes constructed satisfy the mds and
 other properties, the reader is referred to Section \ref{general} and 
 for the decoding algorithms the reader should  consult Section
 \ref{decoding}. 

  The reader might appreciate for comparison the mds codes (Section
  \ref{char2})  of   types \\ $(255,253,3), (255,251,5),  
...,(255,155,101), ..., $ or in general of type $(255,r,256-r)$,  constructed over $GF(2^8)$ together with
decoding algorithms. The methods may be extended to form
mds codes  over $GF(2^s)$ with decoding
algorithms. 
It is shown that codes of the form
$(256,r,257-r)$ may be generated over the prime field $GF(257)$ with
decoding algorithms and these perform better. 

\subsection{To err is ...} If a code is required to correct one error it must have distance $\geq 3$. If the length is also $\leq 3$ then the code is equivalent to a repetition code, one of the form $(3,1,3)$.   

For a code of length $4$ to be 1-error correcting, and not a repetition
code, it must be a $(4,2,3)$
mds code.  
Look at unit-derived codes from  Fourier $4 \ti 4$ matrices for such. No $4\ti 4$ Fourier matrix exists in characteristic $2$ as $2 \vert 4$. 
Consider characteristic $3$. Now $3^2-1 = 8$ so there exists an element
of order $8$ in $GF(3^2)$ and thus an element of order $4$ exists in
$GF(3^2)$. To construct $GF(3^2)$ use a primitive polynomial of degree
$2$ over $\Z_3 = GF(3)$ such as $x^2+x+2$. Then $x $ has order $8$ and
$x^2=\om $ has order $4$. Now form the $4\ti 4$ Fourier matrix $F_4$
over $GF(3^2)$ with $\om$ as the primitive $4^{th}$ root of 1.

By general theory, the first two rows or any two rows in succession
of a Fourier $F_4$ matrix  gives a generator matrix of a $(4,2,3)$ code.
  The rate
of these codes is $\frac{2}{4}=\frac{1}{2}$. 

Row 4 followed by row 1 also
works but note that row 1 with row 3 will not give an mds code. Why? 
    
The order of $GF(5) \backslash 0$ is
$4$. Then it is required to find an element of order $4$ in $GF(5)$ and
it is easily checked that $2$ has order 4 modulo $5$ as $3$. Now
form the Fourier $4 \ti 4$ matrix over $GF(5)$ using $2 \mod 5$ as the
primitive element: 
$F_4 =  \begin{pmatrix} 1 &1 &1 &1 \\ 1 & 2 & 4 & 3 \\ 1 & 4 & 1 &4 \\ 1
	& 3 & 4 & 2 \end{pmatrix}$. If the matrix is over $GF(5)$, the
calculations can all be done with modulo 5 arithmetic.

A length $5$ code could also correct 1 error if it is of the form
$(5,3,3)$. The rate here is $\frac{3}{5}$. What is required is a Vandermonde or
Fourier matrix of size $5 \ti 5$ over a field. Such can be constructed
in $GF(2^4), GF(3^4), ...$ but not in characteristic $5$ of course.  

For a length $6$ code it is required to construct a Vandermonde or
Fourier $6\ti 6 $ matrix and extract codes from the rows using the
unit-derived method. A $(6,2,5)$
code can correct 2 errors but the rate is small. Consider constructing
$(6,4,3)$ codes with 1-error correcting capability and rate $\frac{2}{3}$. 
$GF(7)$ has elements of order $6$ such as  $3$ or $5$ and these can
be used to construct a Fourier $6\ti 6$ matrix over $GF(7) =
\Z_7$. Taking the first four rows or any four rows in succession will
generate a $(6,4,3)$ code over $GF(7)$.  

All the small length codes mentioned here and below may be constructed  directly
using for example a package such as GAP,    
 containing  the coding sub-package GUAVA,  reference \cite{gap}.   

\subsection{Worked example of decoding algorithm}\label{decode} In
Section \ref{general} decoding algorithms are derived. Here an example
of the workings  of the decoding algorithms developed later is given.  

Let  $\F = GF(29)$. A generator of $\{\F\backslash 0\}$ has order
$28$. We are interested in a Fourier $7\ti 7$ matrix over $\F$. An
element of order $7$ is easily obtained in $\F$ and indeed $7^7 \equiv 1 \mod
29$.  

Consider then the unitary scheme:

$$\begin{pmatrix}1&1&1&1&1&1&1 \\ 1 &\om & \om^2 & \om^3 &\om^4 & \om^5 &\om^6 \\ 1 & \om^2 & \om^4 & \om^6 & \om &\om^3 & \om ^5 \\ 1 & \om^3& \om^6& \om^2& \om^5 & \om & \om^4 \\ 1 & \om^4 & \om & \om^5 & \om^2&\om^6&\om^3 \\ 1&\om^5 &\om^3&\om&\om^6&\om^4&\om^2 \\ 1 & \om^6&\om^5&\om^4&\om^3&\om^2&\om \end{pmatrix}
\begin{pmatrix}1&1&1&1&1&1&1 \\ 1 &\om^6 & \om^5 & \om^4 &\om^3 & \om^2 &\om \\ 1 & \om^5 & \om^3 & \om & \om^6 &\om^4 & \om ^2 \\ 1 & \om^4& \om& \om^5& \om^2 &\om^6 & \om^3 \\ 1 & \om^3 & \om^6 & \om^2 & \om^5&\om&\om^4 \\ 1&\om^2 &\om^4&\om^6&\om&\om^3&\om^5 \\ 1 & \om &\om^2&\om^3&\om^4&\om^5&\om^6 \end{pmatrix} = 7I$$

where $\om$ is a primitive $7^{th}$ root of unity. Here we may take $\om = 7 \mod 29$ and powers of $7$ are evaluated $\mod 29$.  Other values for $\om$ are possible and what is required is an element of order $7$ modulo $29$. \footnote{That $\om = 7 \mod 29$ is used here is coincidental to the size of the matrix.} Let the first matrix above be denoted by $P$ and the second by $Q$. Thus $PQ=7*I$ which is the unit scheme $P\{\frac{1}{7}Q\}= I$. Now choose $r$ rows of $P$ to form a matrix which generates a $(7,r)$ code and a check matrix for this code is obtained from $Q$ by eliminating the columns corresponding to the chosen rows of $P$; in theory the check matrix is from $1/7*Q$ but if $H$ is a check matrix then so is $7*H$. 

From $P$ then $(7,3,5)$ and $(7,5,3)$ codes may be obtained by taking in particular the first 3 rows or 5 rows of $P$ or indeed by taking the required number of rows consecutively from $P$. The general theory which verifies this, including the distances obtained,  is given  in Section \ref{general} below. 

A $(7,5,3)$ code is 1-error correcting. Take the first 5 rows of $P$ as the generator matrix $A$ and then the last two columns, $D$, of $V$ is the check matrix. 
A codeword is $\al A$ for a $1\ti 5$ vector $\al$. Suppose $\al A+ \ep $ is received where $\ep$ is the error and has just one non-zero entry. Applying $D$ to $\al A + \ep $ gives $\ep D$. Now $\ep D$ is a multiple of a row of $D$ as $\ep$ has only one non-zero entry, and this uniquely defines the row and its multiple. 
Thus the error $\ep$ may be eliminated. When the error has been eliminated, then $\al A * C =7*\al$ decodes the word where $C$ denotes the first 5 columns of $Q$.

This decoding method of identifying the multiple of the row of the check
matrix works whenever just 1-error needs correcting. 

A 2-error correcting code $(7,3,5)$ is  obtained from this unit scheme
by taking any three rows of $P$ as a generator matrix. The code may be
corrected as following; the details of the algorithm may be found in
\cite{hurleycomp} which was derived from the error-correcting methods of
Pellikaan \cite{pell}. The algorithm utilises error-correcting pairs
which are shown to exist for these codes. 

Suppose the first 3 rows are the generator matrix of a code $\mathcal{C}$. Then the last 4 columns of $Q$ constitute a check matrix. Let these columns 
 be denoted by $\{E_4\T,E_3\T,E_2\T,E_1\T \}$ in order. Then $\mathcal{C}\T$ is generated by these columns, written as rows. The first three rows of $P$ are $\{ E_0,E_1,E_2\}$ where $E_0$ consists of all $1^s$. 

Now by \cite{hurley} and \cite{pell} an error-correcting pair for $\mathcal{C}$ is as follows:

 $U= \langle E_1,E_2,E_3 \rangle, V=\langle E_0,E_1 \rangle $ are error correcting pairs for $\mathcal{C}$. 


Let $\al A$ be the codeword but when transmitted an error is introduced and the word received is $\al A + w$. Note $w$ is a $1\ti 7$ vector. 
Apply the check matrix which has columns $\{E_4\T,E_3\T, E_2\T,E_1\T\}$ and then $<w,E_i> = w E_i\T=E_i w\T$ are known for $i=1,2,3,4$ where $<,>$ denotes inner product. Let  $<w,E_1> = \al_1, <w,E_2>=\al_2,<w,E_3>=\al_3, <w,E_4>=\al_4$. The algorithm then is:

\begin{enumerate} \item Find an element $x\T$ in the kernel of $\begin{pmatrix} \al_1 & \al_2 & \al_3 \\ \al_2 & \al_3 & \al_4 \end{pmatrix}$. Any non-zero element of the kernel will do.
\item Form $ \underline{a} =(E_1, E_2, E_3)x\T$.  
\item Find the locations of the zero coefficients of $\underline{a}$. Say these are at $j_1,j_2$ for $1\leq j_1,j_2 \leq 7$. 
\item Solve $\begin{pmatrix} E_{1,j_1} & E_{1,j_2} \\  E_{2,j_1} & E_{2,j_2} \\  E_{3,j_1} & E_{3,j_2} \\  E_{4,j_1} & E_{4,j_2} \end{pmatrix} \begin{pmatrix} x_1 \\ x_2 \end{pmatrix} = \begin{pmatrix} \al_ 1 \\ \al_2 \\ \al_3 \\ \al_4 \end{pmatrix}$. Here $E_{k,l}$ denotes the $l^{th}$ entry of $E_k$.
\item $w$ is then $x_1$, located at $j_1$,  and $x_2$, located at $j_2$, and zeros elsewhere. \end{enumerate}  

Suppose now  that $\om =7 \in GF(29)$ is taken as the $7^{th}$ root of unity of the Fourier matrix and the  $\al_i$ are found to be: 
$\al_1=18, \al_2=15, \al_3=4,\al_4=12$.
 Then \begin{enumerate} 
\item An element in $\ker \begin{pmatrix} 18 & 15 & 4 \\ 15 & 4 & 12\end{pmatrix}$ is $x\T=(23,5,1)\T$
\item  $ \underline{a} =(E_1, E_2, E_3)x\T = (0,24,20,1,0,2,11)$. This
       has zeros at positions $j_1=1,j_2=5$.
\item Solve  $\begin{pmatrix} E_{1,j_1} & E_{1,j_2} \\  E_{2,j_1} &
	      E_{2,j_2} \\   E_{3,j_1} & E_{3,j_2} \\  E_{4,j_1} &
	      E_{4,j_2} \end{pmatrix} \begin{pmatrix} x_1 \\ x_2
				      \end{pmatrix} = \begin{pmatrix}
						      \al_ 1 \\ \al_2 \\
						      \al_3 \\ \al_4
						      \end{pmatrix}$ is
		   then solve  $\begin{pmatrix} 1 & 23 \\  1 & 7 \\  1 &
				16 \\  1 & 20 \end{pmatrix}
      \begin{pmatrix} x_1 \\ x_2 \end{pmatrix} = \begin{pmatrix} 18 \\ 15
						 \\ 4 \\  12
						 \end{pmatrix}$. This
      has solution $x_1=1, x_2 =2$. 
\item Then the error is $x_1$ located at $j_1=1$ position  and $x_2$ located at position $j_2=5$ giving the error vector $w=(1,0,0,0,2,0,0)$.

\end{enumerate}

The calculations in this case are all done in $\Z_{29} =GF(29)$. 

\subsection{Further samples} 
\subsubsection{$11 \ti 11$ cases}

Suppose  a Vandermonde or Fourier $11\ti 11$ matrix $F_{11}$ over a field $\F$ has been  found. 
Now choose rows consecutively\footnote{Other choices are possible.} to construct codes, and error-correcting pairs exist for these codes. In Section 
\ref{general} below it is shown that such codes from $F_{11}$ are mds, maximal distance separable codes and decoding methods are derived in Section \ref{decoding}.   

Thus $(11,3,9)$ codes which have 4-error correcting capability, $(11, 5,7)$ which have 3-error capability, $(11,7,5)$ which have 2-error correcting capability, and $(11,9,3)$ which have 1-error capability ability are obtained. The decoding algorithms reduces to finding $t$-error correcting pairs. 


An example of such a field  which has an easily workable $11^{th}$ of
unity is $GF(23)$. The group of non-zero elements in $GF(23)$ is of
order $22$ and is cyclic so elements of order $11$ exist. In fact $2
\mod 23$ or $3 \mod 23$
have order $11$ in $GF(23)$ and either of these 
may be used as a primitive $11^{th}$ root of unity in forming $F_{11}$. 
The calculations   in this case are arithmetic  modulo $23$. 

Let $F_{11}$ denote the Fourier matrix in $GF(23)$ with $\om=2 \mod 23$ as the
primitive $11^{th}$ root of unity. 
Take consecutive rows or else select rows in arithmetic sequence of
their order.  An efficient  decoding algorithm  using error correcting pairs
exists for these codes is given generally in Section \ref{decoding}; the algorithm  is derived from \cite{hurleycomp}. 

Notice that $11$ divides $2^{10}-1$ so the Fourier matrix of size $11
\ti 11$ can also be constructed over $GF(2^{10})$. However this field is
large and calculations may be  more difficult. But see Section
\ref{char2} below for discussion of characteristic $2$ cases which have
other advantages.  

Note that  $11$ divides $3^5-1$ so the field $GF(3^5)$ could also be used.

\subsubsection{$13 \ti 13 $ cases}
For $13 \ti 13$ Fourier matrices there are a number of possibilities. To
work  in modular arithmetic take $\F = GF(53)$ as $13$ divides
$(53-1)=52$, and then there exists primitive $13^{th}$ of unity. In fact
$10^{13}\equiv 1 \mod 53$ so $10 \mod 53$ may be used as the primitive $13^{th}$
root of unity in $GF(53)$ in forming the Fourier $13 \ti 13$ matrix. 

In $GF(3^3)$ also there exists a $13^{th}$ root of unity as $3^3-1=26 =
2*13$. So indeed the square of the generator of the non-zero elements of
$GF(3^3)$ is a primitive $13^{th}$\footnote{Note that
$13$ is a base $3$ repunit.} root of unity.  Use  an irreducible primitive polynomial of
degree $3$ in $\Z_3=GF(3)$ with which the calculations may be made in
$GF(3^3)$.    

\subsection{Characteristic 2 cases}\label{char2}

Characteristic $2$ cases are always interesting and this is indeed the case with these unit-derived codes from Vandermonde/Fourier matrices. 

Codes over $GF(2^s)$ may be transmitted as binary signals. 
 The code symbols are within $GF(2^s)$. If each code symbol is represented by an $s$-tuple over $GF(2)$, then the code can be transmitted using binary signalling. In decoding, every $s$ received bits are grouped into a received signal over $GF(2^s)$. 

\begin{enumerate} 

\item As $2^2-1=3$ so $3\ti 3$ Fourier matrices over
			$GF(2^2) $ can be obtained and mds codes may be
			derived from this. These however are equivalent
			to repetition codes $(3,1,3)$ or to codes of the
			form $(3,2,2)$ which do not have
			error-correcting capabilities. 
\item $2^3-1 = 7$ gives a Fourier $7\ti 7$ matrix over $GF(2^3)$. Thus codes $(7,3, 5)$ which are $2$-error correcting and codes $(7,5,3)$ which are 1-error correcting may be formed over $GF(2^3)$.
\item $2^4-1=15$ and so $(15,13,3), (15, 11,5), (15,9,7), (15, 7,9)$ codes can be formed by this method over $GF(2^4)$. 
\item $2^5-1=31$, which is prime, enables $(31, 29, 3), (31,27,5), (31,25,7), (31,23,9) ,....$  codes to be formed over $GF(2^5)$. If rate about $3/4$ is required then take $(31, 23,9)$ which is $4$-error correcting. 
\item $2^6-1=63$. Codes of form $(63,r, 64-r)$ may be formed with efficient 
 error-correcting algorithms. 
\item $2^7-1=127$. Fourier $127\ti 127$ matrices may be formed over $GF(2^7)$. Note that $127$ is prime, in fact a Mersenne prime, and Fourier matrices of length a Mersenne prime are interesting. Here mds codes of form $(127, 125, 3), (127,123,5), ...,(127, 87,41),...., $ may be formed using unit-derived codes from this Fourier matrix over $GF(2^7)$.  Note for example that $(127,97,31)$ has rate 
$\frac{97}{127} > \frac{3}{4}$ and can correct $15$ errors. 

Use a prime field? From  the prime field $GF(127)$ a Fourier $126 \ti 126$ matrix may be formed with elements from  $GF(127)=\Z_{127}$ and unit-derived codes may be constructed from this; the algebra then is $\mod 127$. 
\item Now $2^8-1 =255$ and this is an interesting case as mds codes over
      $GF(2^8)$ are in practical use. The Reed-Solomon (see for example
      \cite{blahut}),  $(255, 239,17)$ code over $GF(2^8) $ is used
      extensively in data-storage systems, hard-disk drives and optical
      communications; the Reed-Solomon $(255, 223,33)$ code  over
      $GF(2^8)$ is or was the NASA standard for deep-space and satellite
      communications.  

Form the Fourier $255\ti 255$ matrix using a primitive $255^{th}$ root of unity in $GF(2^8)$. A primitive polynomial of degree $8$ over $\Z_2 = GF(2)$ would be useful here; lists of these are known and one such is  $x^8+x^4+x^3+x^2+1 $. By taking unit-derived codes from this Fourier matrix one readily gets $(255, 253,3),(255,251,5), ..., (255, 239,17), ...,(255,223, 23), ..., (255,155,101),...$ codes. So for example the code $(255, 155, 101)$ can correct $50$ errors. Practical error-correcting algorithms for these are given within the general form of 
  Section \ref{general}.  

A better way perhaps of constructing these types of codes is to consider
      the prime $257$ and then the field $GF(257)$. The order of the
      units of $GF(257)$ is $256$ and then construct the Fourier $256\ti
      256$ matrix over $GF(257)$ using a primitive $256^{th}$ root of
      unity. Now the order of $3 \mod 257$ is $256$ so indeed $3\mod
      257$ could be used as
      this primitive root of unity in forming the Fourier $256\ti 256$
      matrix over $GF(257)$.  Other primitive
      generators could be used such as $5$ as the order of $5\mod 257$
      is also $256$.   Note here also that the arithmetic is
      modular arithmetic in $\Z_{257}=GF(257)$. For example codes of form
      $(256,222, 35)$ with efficient decoding algorithm which can
      correct $17$ errors may be formed over $GF(257)$; indeed codes
      of the form $(256,r,257-r)$ may be formed over $GF(257)$ 
with efficient decoding algorithms for $1\leq r \leq n$. 

\item Clearly also one can go much further and work with $GF(2^s)$ for $s> 8$.

\end{enumerate}


\subsection{Using special fields}\label{further} 


Suppose we require that the Fourier matrix, from which the unit-derived codes are generated,  be of size $p\ti p$ for a prime $p$. 

\subsubsection{Mersenne and repunit primes } 
Fields of characteristic $2$ were considered in section \ref{char2}. 

Suppose  the generator of the non-zero elements of 
$GF(2^s)$ is  of  order a  prime $p$ and form  
the Fourier $p\ti p$ matrix using this
generator as the $p^{th}$ root of unity. This gives a $p\ti p$ matrix 
over $GF(2^s)$ from which unit-derived mds codes may be generated; these have nice properties.  For example when rows are selected in arithmetic sequence $k$ then always $\gcd(n,k)=1$ and the resulting codes have  efficient decoding algorithms. 

Saying the non-zero elements of $GF(2^s)$ have order a prime is simply saying  that $2^s-1$ is a Mersenne
prime. The first Mersenne  primes are $3,7,31, 127, .., $ but it is unknown if there are an infinite number of these.

The fields $GF(2^5)$, $GF(2^7)$ with  
$2^5-1=31$ and $2^7-1=127 $ were given as examples in Section \ref{char2}. 

All these have efficient error-correcting
algorithms as explained in Section \ref{decoding}. 

One can also consider {\em repunit base $p$ primes}. Now $q$ is a
repunit base $p$ prime if $q$ is a prime and $p^s-1=(p-1)q$ for some
$s$. Repunit base $2$ primes are the Mersenne primes. Using repunit base
$p$ prime $q$ with $p^s-1=(p-1)q$ leads to considering $q\ti q$ Fourier
matrices over $GF(p^s)$. Details are omitted. 

\subsubsection{Germain primes}
It is often useful to have a prime size Fourier matrix in as small a field as 
possible. If this field is also a prime field, then this is even better as the calculations are  then modular arithmetic over the prime field.  
Thus we are lead to consider 
    {Germain primes}.  Now $p$ is a {\em Germain prime} if $2p+1$ is also a prime. 

Consider the field $GF(2p+1)$ where $p$ is also a prime. 
A generator $\om $ of the non-zero elements of $GF(2p+1)$ has order
$2p$ and thus $\al=\om^2$ has order $p$.  Now form the Fourier $p\ti p$
matrix over $GF(2p+1)$ using $\al$ as a primitive $p^{th}$ root of
unity. Codes are then formed from the rows of this
Fourier matrix and these are mds codes with efficient decoding
algorithms. As the codes are over $GF(2p+1)$ the arithmetic is modular
arithmetic over $\Z_{2p+1}$. 

The first Germain primes are $2,3,5,11,23, 29, 41, ...$. 

For example $p= 29$ gives $2*p+1=59$ and form a Fourier $29\ti 29$
matrix over $GF(59)$ using the square of any generator of the non-zero
elements of $GF(59)$. The order of $2\mod 59$ is $58$ so the order $4 \mod
59$ is $29$;  however the order of $3 \mod 59$ is also $29$ and this is preferable. Thus take $\om =
3 \mod 59$ and form the Fourier $29\ti 29$ matrix over $GF(59)$ using
this $\om$ as the primitive $29^{th}$ root of 1. 

\section{General enabling results}\label{general} 
In \cite{hurleycomp} conditions are given to ensure that subdeterminants
of Vandermonde matrices are non-zero. Fourier matrices are special types
of Vandermonde matrices. Such conditions can  be applied to  generate   
codes from units with maximum possible distance and further it is  shown that
practical decoding algorithms for these codes exist.  

Of particular relevance  in \cite{hurleycomp} is Section 6, noting  Proposition 6.1 and its corollaries. 

\subsection{Determinants of submatrices}
The Vandermonde matrix $V=V(x_1,x_2,\ldots,x_n)$ is defined by
   
$V=V(x_1,x_2,\ldots,x_n) = \begin{pmatrix}
1&1&\ldots &1 \\ x_1& x_2& \ldots& x_n \\ \vdots & \vdots &
\vdots & \vdots \\ x_1^{n-1} & x_2^{n-1} & \ldots &
x_n^{n-1} \end{pmatrix}$ 

It is assumed that entries of a Vandermonde matrix here are over a 
field and not necessarily over the real or complex numbers.  
It is well-known that the determinant of $V$ is non-zero if and only
if the $x_i$ are distinct; in fact $\det V = \prod_{i<
j}(x_i-x_j)$. 

Assume in addition from now on that {\em all entries of a Vandermonde
matrix used here are non-zero}.  

The following Proposition and its corollaries are taken from
\cite{hurleycomp}. The proofs are
included again here for completeness and for their importance. 
\begin{proposition}\label{van1} Let $V=V(x_1,x_2,\ldots, x_n)$ be a Vandermonde
 matrix with  rows and columns numbered $\{0, 1, \ldots, n-1\}$. 
 Suppose rows $\{i_1,i_2,\ldots,i_s\}$ (in order) and columns
 $\{j_1,j_2,\ldots, j_s\}$ are chosen to form an $s\times s$ submatrix $S$
 of $V$ and that $\{i_1,i_2, \ldots, i_s\}$ are in arithmetic
 progression with arithmetic difference $k$. Then

$$|S|= x_{k_1}^{i_1}x_{k_2}^{i_1} \ldots
 x_{k_s}^{i_1}|V(x_{k_1}^k,x_{k_2}^k, \ldots, x_{k_s}^k)|$$
\end{proposition}
\begin{proof} Note that $i_{l+1}-i_l= k$ for $l=1,2,\ldots, s-1$, for
 $k$ the fixed arithmetic difference.

Now $ S=\begin{pmatrix}x_{k_1}^{i_1} & x_{k_2}^{i_1} & \ldots &
       x_{k_s}^{i_1} \\ x_{k_1}^{i_2} & x_{k_2}^{i_2}&\ldots
       & x_{k_s}^{i_2}
\\ \vdots & \vdots & \vdots & \vdots \\ x_{k_1}^{i_s} & x_{k_2}^{i_s}&
       \ldots & x_{k_s}^{i_s} \end{pmatrix} $
 and so $|S| = \left|\begin{array}{cccc}x_{k_1}^{i_1} & x_{k_2}^{i_1} & \ldots &
       x_{k_s}^{i_1} \\ x_{k_1}^{i_2} & x_{k_2}^{i_2}&\ldots
       & x_{k_s}^{i_2}
\\ \vdots & \vdots & \vdots & \vdots \\ x_{k_1}^{i_s} & x_{k_2}^{i_s}&
       \ldots & x_{k_s}^{i_s} \end{array}\right| $.

Hence by factoring out $x_{k_i}$ from column $i$ for $i=1,2,\ldots, s$
it follows that  

$|S| = x_{k_1}^{i_1}x_{k_2}^{i_1}\ldots x_{k_s}^{i_1} \left|
 \begin{array}{cccc}1&1&\ldots &1 \\ x_{k_1}^k & x_{k_2}^k & \ldots & x_{k_s}^k
\\ x_{k_1}^{2k} & x_{k_2}^{2k} & \ldots &x_{k_s}^{2k} \\ \vdots &
  \vdots & \vdots & \vdots \\ x_{k_1}^{(s-1)k} & x_{k_2}^{(s-1)k} &
  \ldots & x_{k_s}^{(s-1)k}\end{array}\right|
= x_{k_1}^{i_1}x_{k_2}^{i_2}\ldots x_{k_s}^{i_s}|V(x_{k_1}^k,x_{k_2}^k,
 \ldots, x_{k_s}^k)|$

\end{proof}

A similar result holds when the columns $\{j_1,j_2,\ldots,
j_s\}$ are in arithmetic progression. 
\begin{corollary} $|S|\neq 0$ if and only if $|V(x_{k_1}^k,x_{k_2}^k,
 \ldots, x_{k_s}^k)|\neq 0 $. 
\end{corollary}
\begin{corollary} $|S| \neq 0$ if and only if $x_{k_i}^k \neq x_{k_j}^k$
 for $i\neq j, 1\leq i,j \leq s$. This happens if and only if
 $(x_{k_i}{x_{k_j}^{-1}})^k \neq 1$ for $i\neq j, 1\leq i,j \leq s$. 
\end{corollary}
\begin{corollary}\label{unity}  $|S|\neq 0$ if and
 only if $(x_{k_i}x_{k_j}^{-1})$ is not a $k^{th}$ root of unity for
 $i\neq j, 1\leq i,j\leq s$.
\end{corollary}
\begin{corollary}\label{jump} When $k=1$ (that is when consecutive rows are taken) then $|S| \neq 0$.
\end{corollary}
\begin{proof} This follows from Corollary \ref{unity} as $(x_{k_i}x_{k_j}^{-1})\neq 1$ for $i\neq j$.
\end{proof} 
\begin{corollary}\label{fouri} Let  $x_i=\om^{i-1}$ where $\om$ is  a primitive
 $n^{th}$ root of unity (that is, when $V$ is the Fourier
  $n\times n$ matrix) and suppose $\gcd(k,n)=1$. Then $|S|\neq 0$. 
\end{corollary}
\begin{proof} If $(x_{k_1}x_{k_j}^{-1})^k= 1$ then
  $(\om^{k_i-1}\om^{1-k_j})^k=1$ and so $\om^{k(k_i-k_j)} = 1$. 
  As $\om$ is a primitive $n^{th}$ root of unity this implies 
    that $k(k_i-k_j) \equiv 0 \mod n$. As $\gcd(k,n)=1$ this implies
  $k_i-k_j \equiv 0 \mod n$ in which case $k_i=k_j$ as $1\leq k_i<n,
  1\leq k_j < n$. 
\end{proof}

  
Recall that an mds code is one of the form $(n,r,n-r+1)$ which attains the maximum distance possible for an $(n,r)$ code.  mds codes with efficient decoding algorithm are  the goal. 

An mds $(n,r)$ code $\mathcal{C}$ is characterised by either of the
following equivalent conditions, \cite{blahut}:
\begin{itemize}
\item $\mathcal{C}$ is an $(n,r,n-r+1)$ code. 
\item $\mathcal{C^\perp}$ is an mds $(n,n-r, r+1)$ code, where $\mathcal{C^\perp}$ is the dual code of $\mathcal{C}$. 
\item Any $(n-r)$ columns of a check matrix for  $\mathcal{C}$ are linearly independent.
\item  Any $r$ columns  of a generator matrix for $\mathcal{C}$ are linearly independent.  

\end{itemize} 


As long as we take the rows of the $n\ti n$ Vandermonde 
matrix in arithmetic sequence $k$ and the entries
$x_i$ are such that $(x_ix_j^{-1})$ is not a $k^{th}$ root of unity for
$i\neq j$ then mds codes will be generated by these rows. 
When $k=1$,  in which case consecutive rows of the matrix are taken, then always $\gcd(n,k)=1$.  When the Vandermonde matrix in question is the Fourier matrix  in addition it will be shown that practical decoding algorithms exist for these cases. 



\subsection{Fourier matrix} The Fourier matrix is a special type of
Vandermonde matrix. Let $\om$ be a primitive $n^{th}$ root of unity in a
field $\F$. The Fourier matrix $F_n$, relative to $\om$ and $\F$, is the $n\ti n$ matrix  

$$ F_n= \begin{pmatrix}1 & 1 & 1& \ldots & 1 \\ 1 & \om & \om^2 & \ldots &
	 \om^{n-1} \\ 
1 & \om^2 & \om^4 & \ldots & \om^{2(n-1)} \\ \vdots & \vdots & \vdots &
    \ldots & \vdots \\ 1 & \om^{n-1} & \om^{2(n-1)} & \ldots &
    \om^{(n-1)(n-1)} \end{pmatrix}$$

Simplifications can be made to the powers by noting $\om^n=1$. 

Then  $$\begin{pmatrix}1 & 1 & 1& \ldots & 1 \\ 1 & \om & \om^2 & \ldots & 
	 \om^{n-1} \\ 
1 & \om^2 & \om^4 & \ldots & \om^{2(n-1)} \\ \vdots & \vdots & \vdots &
    \ldots & \vdots \\ 1 & \om^{n-1} & \om^{2(n-1)} & \ldots &
    \om^{(n-1)(n-1)} \end{pmatrix} \begin{pmatrix}1 & 1 & 1& \ldots & 1 \\ 1
				    & \om^{n-1} & \om^{2(n-1)} & \ldots
				    & \om^{(n-1)(n-1)}\\ 
1 & \om^{n-2} & \om^{2(n-2)} & \ldots &\om^{(n-1)(n-2)} \\ \vdots & \vdots & \vdots &
    \ldots & \vdots \\ 1 & \om & \om^{2} & \ldots &
    \om^{(n-1)} \end{pmatrix} = nI_n$$

The inverse of $F_n$ can be obtained from the above by multiplying
through by $n^{-1}$ when it exists. An $n^{th}$ root of unity can only
exist in a field provided the characteristic of the field does not
divide $n$ and in this case the $n^{-1}$ exists.


If $\om$ is a primitive $n^{th}$ root of unity then so is $\om^k$
where $\gcd(n,k) = 1$ and in these cases the Fourier matrix may be
defined by replacing $\om$ by $\om^k$ to obtain another Fourier
matrix. Notice that the second matrix on the left in the above is
obtained by replacing $\om$ by $\om^{n-1}$ and is thus also a Fourier
matrix (relative to $\om^{n-1}$ and $\gcd(n,n-1) = 1$). 

Denote the  rows of $F_n$ in order by $\{E_0, E_1, \ldots,
E_{n-1}\}$. It is easily checked that $E_iE_{n-i}\T=n$ and $E_iE_j\T = 0$
for $j\neq n-i \mod n$. Thus 

$$\begin{pmatrix}E_0 \\ E_1 \\ \vdots \\ E_{n-1}\end{pmatrix} (E_0\T,
E_{n-1}\T, E_{n-2}\T, \ldots, E_1\T) = nI_n$$

Call this the {\em Fourier Equation} for future reference. We are
assuming the Fourier matrix exists over the field and in particular 
 any $r$ rows or
any $r$ columns are linearly independent. 

Suppose then the first $r$ rows of $F_n$ are used to form a generating
matrix $A$ for a $(n,r)$ code $\mathcal{C}_r$. Now using the unit-derived scheme from the Fourier matrix we see that  

$$\begin{pmatrix}E_0 \\ E_1 \\ \vdots \\ E_{r-1}\end{pmatrix} 
(E_{n-r}\T, E_{n-2}\T, \ldots, E_1\T) = 0_{n-r}$$ 

which corresponds to $AD=0_{n-r}$ 
where $D\T$ is a  check matrix.   Thus a check matrix is 

$\begin{pmatrix}E_{n-r} \\ E_{n-r-1} \\ \vdots \\ E_1 \end{pmatrix} $
 and hence  
$\begin{pmatrix} E_1 \\  E_2 \\ \vdots \\ E_{{n-r}}\end{pmatrix}$ 
is a check matrix. 
 
Suppose a codeword $\al A$ is transmitted but $\al A + w$ with error $w$ is received where $w$ is an $1\ti n$ vector. Then  $<E_i,w> = \al_i$ are known for $i=1,2, \ldots (n-r)$ since  $(\al A + w)E_i\T= w E_i\T= <w,E_i>$ for these $i$.

The star 
multiplication, $*$ ,  is explained further 
in Section \ref{prelim1} but  is simply multiplying corresponding
entries of vectors: If $x_i$
denotes the $i^{th}$ component of a vector $\underline{x}$ in $\F^n$
then $\underline{a}*\underline{b}$ for $\underline{a},\underline{b}\in
\F^n$ is defined to be the vector with components $a_i*b_i$ in $i^{th}$ position. 
The rows of $F_n$ also have the nice property that $E_i*E_j=E_{i+j}$
where suffices are taken $\mod n$ and this is very useful for describing
error-correcting algorithms. 


\subsection{Consecutive rows}
First of all consider cases where consecutive 
 rows of the Vandermonde matrix are taken to define a unit-derived code.

The Vandermonde matrix is  

$V=V(x_1,x_2,\ldots,x_n)  = \begin{pmatrix}
1&1&\ldots &1 \\ x_1& x_2& \ldots& x_n \\ \vdots & \vdots &
\vdots & \vdots \\ x_1^{n-1} & x_2^{n-1} & \ldots &
x_n^{n-1} \end{pmatrix}$ 
 
This has inverse $U$ with $VU=I_n$. When $V$ is a Fourier matrix the
inverse matrix $U$ of $V$ is easy to find and can be written down directly. 

Let $A$ be the matrix of the first $r$ rows of $V$ and
 $D$ the matrix of the last $(n-r)$ columns of $U$. By unit-derived
 scheme then, $AD=0$ and 
 $D\T$ is the check matrix of the $(n,r)$ code $\mathcal{C}$ generated by
  $A$. Now $\mathcal{C}^\perp$ is the dual code of
 $\mathcal{C}$ and is generated by the rows of $D\T$. It is known that 
 $\mathcal{C}$ is
 an mds code if and only if $\mathcal{C}^\perp$ is an mds code. 

\begin{proposition} Any $r\ti r$ submatrix of $A$ is a Vandermonde
 matrix $V(x_{i_1}, x_{i_2}, \ldots, x_{i_r})$ for $i_j \in \{1,2,\ldots,n \}$ with $i_1<i_2 < \ldots < i_r$. 
\end{proposition} 

\begin{proof} This follows from Proposition \ref{van1} above. 


\end{proof}
\begin{corollary} Any $r\ti r$ submatrix of $A$ has $\det \neq 0$. 
\end{corollary}
\begin{corollary} The code $\mathcal{C}^\perp$ is an mds code. 
\end{corollary}
\begin{proof} This is true since $A$ is the check matrix of
 $\mathcal{C}^\perp$ and every $r\ti r$ submatrix of $A$ has non-zero
 determinant so that the minimum distance of the $(n,n-r)$ code
 $\mathcal{C}^\perp$ is $r+1$.
\end{proof}  
\begin{corollary} The code $\mathcal{C}$ is an $(n,r,n-r+1)$ mds code. 
\end{corollary}
\begin{proof} This is because $\mathcal{C}^\perp$ is an mds
 $(n,n-r,r+1)$ code. It may also be seen from the fact that any $r$
 columns of $A$ are linearly independent since the determinant of any
 $r\ti r$ submatrix of $A$ is $\neq 0$. 
\end{proof}

\quad 

Take any  $r$ consecutive rows of a Vandermonde matrix as follows:
 
$A=\begin{pmatrix} x_1^{r_1} & x_2^{r_1} & \ldots & x_n^{r_1} 
\\ x_1^{r_1+1}& x_2^{r_1+1}& \ldots & x_n^{r_1+1} \\ \vdots & \vdots &
\vdots & \vdots \\ x_1^{r_1+ r-1} & x_2^{r_1+r-1} & \ldots &
x_n^{r_1+r-1} \end{pmatrix}$ 

Write $i_j$ for $x_{i_j}$. Now  any $r\ti r$ submatrix of $A$ has the form

$\begin{pmatrix}
 i_1^{r_1} & i_2^{r_1} & \ldots & i_r^{r_1} \\ \vdots & \vdots &
\vdots & \vdots \\ i_1^{r_1+r-1} & i_2^{r_1+r-1} & \ldots &
i_r^{r_1+r-1} \end{pmatrix}$.

The determinant of this is by Proposition \ref{van1}

$i_1^{r_1}i_2^{r_1}\ldots i_r^{r_1} \left|\begin{matrix}
1&1&\ldots &1 \\ i_1 & i_2 & \ldots & i_r \\ \vdots & \vdots &
\vdots & \vdots \\ i_1^{r-1} & i_2^{r-1} & \ldots &
i_r^{r-1} \end{matrix}\right| = i_1^{r_1}i_2^{r_1}\ldots i_r^{r_1} |V(i_1,i_2, \ldots, i_r)| $
 
This is clearly non-zero - we are assuming the $x_j$ are distinct and non-zero. 

This gives further  mds codes from the unit scheme. 

\begin{proposition} Let $\mathcal{C}_r$ be a code obtained by taking any $r$ rows in succession of a Vandermonde $n\ti n$ matrix as  a generator matrix. Then $\mathcal{C}_r$ is an mds $(n,r,n-r+1)$ code.
\end{proposition} 

\subsection{Rows in arithmetic sequence}
Now  choose  $r$ rows in sequence with the same arithmetic difference
$p$. Consider the case where the sequence starts at the first row; cases
where the sequence begins  at another row are similar.  Then the matrix formed is  

$A=\begin{pmatrix}
1&1&\ldots &1 \\ x_1^p& x_2^p& \ldots& x_n^p \\ x_1^{2p}& x_2^{2p}& \ldots& x_n^{2p} \\ \vdots & \vdots &
\vdots & \vdots \\ x_1^{p(r-1)} & x_2^{p(r-1)} & \ldots &
x_n^{p(r-1)} \end{pmatrix}$ .

Here we begin at the first row and assume $p(r-1) \leq n$.   It may be possible to overlap and take $p*j $ to be $p*j \mod n$,  and the added assumption that $r
< n$. In particular  overlapping is possible when the Vandermonde unit schemes consist of Fourier matrices. 

The check matrix is obtained by deleting the corresponding columns of the inverse of $V$. 

Any $r\ti r$ submatrix of $A$ has the form 

$\begin{pmatrix}
1&1&\ldots &1 \\ i_1^p & i_2^p & \ldots & i_r^p \\ i_1^{2p} & i_2^{2p} & \ldots & i_r^{2p} \\ \vdots & \vdots &
\vdots & \vdots \\ i_1^{p(r-1)} & i_2^{p(r-1)} & \ldots &
i_r^{p(r-1)} \end{pmatrix}$

where $i_j^k$ means $x_{i_j}^k$. 
This has determinant $\prod_{k<j} (i_k^p - i_j^p)$. It is easy to decide  when this is non-zero. 

This determinant is non-zero if and only for all $i_k,i_j, k\neq j $ that $i_k^p - i_j^p \neq 0$ and this happens if and only if $(i_ki_j^{-1})^p \neq 1 $ which happens if and only if  $i_ki_j^{-1}$ is not a $p^{th}$ root of unity. 

From Corollary \ref{fouri} it is noted that when $\gcd(n,k) = 1$ and the Vandermonde matrix is a Fourier matrix then the determinant is never $0$. This gives the following proposition. 

\begin{proposition}\label{fourier} Let $F$ be a Fourier $n\ti n$ matrix. Suppose a code is obtained from $F$ by choosing in order $r$ rows which are in arithmetic sequence $k$ with $\gcd(n,k)=1$ to form the generator matrix of a code. Then the code is an mds $(n,r,n-r+1)$ code.
\end{proposition}

Note also for the Fourier matrix that it is possible to {\em overlap} in selection  and still obtain an mds code.  

\begin{proposition}\label{vanderm} Let $V=V(x_1,x_2,\ldots,x_n)$ be a Vandermonde  $n\ti n$ matrix such that $x_ix_j^{-1}$ is not a $k^{th}$ root of unity for any $i\neq j$. 
Suppose a code is obtained by choosing in order $r$ rows from $V$  
which are in arithmetic sequence $k$ to form a code. Then the code is an mds $(n,r, n-r+1)$ code.
\end{proposition} 

\section{Decoding}\label{decoding} 


The following decoding methods  are sourced from \cite{hurleycomp} which is an
application of Pellikaan's decoding method using  error correcting pairs
\cite{pell} when such exist. 

Error correcting pairs were introduced by  Pellikaan \cite{pell} and
Duursma \& K\"otter \cite{duur}. The method of Pellikaan is found  more
useful here and in \cite{hurleycomp};  the decoding algorithm of
Pellikaan has a precise translation into a linear algebra method for the
codes constructed here 
 as
explained in Section 3 of \cite{hurleycomp}.     

\subsection{Preliminaries}\label{prelim1}

First some preliminaries are required.
Let $F$ be a field and $\C$ a (linear) code over $F$. Write 
$n(\C)$ for the code length of $\C$, its minimum distance is denoted 
by $d( \C )$ and denote its dimension by $k ( \C )$. 

Now $w_i$ denotes the $i^{th}$ component
of $w\in F^n$. 
For any  $w \in F^n$  define the support of $w$
by $\supp( w ) = \{ i | w_i \neq  0 \}$
and the zero set of $w$ by $z ( w ) = \{ i | w_i = 0 \}$ . The weight of $w$ is
the number of non-zero coordinates of $w$ and denote it by $wt ( w )$. The
number of elements of a set $I$ is denoted by $| I |$. Thus $wt(a)=
|\supp(w)|$.

We say that $w$ has $t$
errors supported at $I$ if $w = c + e$ with $c \in \C$ and $I = \supp( e
)$ and $| I | = t = d ( w , \C )$. 

The bilinear form $< , >$ is
defined by $< a , b > = \sum_i a_i b_i$. For a subset $C$ of $F^n$, 
the dual $C^\perp$ of $C$ in $F^n$ with respect to the bilinear form $< ,
>$ is defined by $C^\perp= \{ x | < x , c > = 0,   \forall c \in C \}$.

As usual the sum of two elements of $F^n$ is defined by adding corresponding 
coordinates.  Of use in these considerations is what is termed 
the {\em star multiplication} $a * b$ of two elements $a, b \in F^n$
defined by multiplying corresponding coordinates, that is $( a* b )_i = a_i
b_i$. For subsets $A$ and $B$
of $F^n$ denote the set $\{ a * b | a \in A, b \in B \}$ by $A * B$. If
$A$ is generated by $X$ and $B$ is generated by $Y$ then $A*B$ is
generated by $X*Y$. 

\begin{Definition}\label{error}
Let $A, B$ and $C$ be linear codes in $F^n$. We call $( A, B
)$ a {\em $t$-error correcting pair for $C$} if 
\\ 1) $A * B \subseteq C^\perp $ \\ 2) $k ( A ) > t$ \\  3) $d
( A ) + d ( C ) > n$, \\ 4) $d ( B^\perp) > t$. 

\end{Definition}

For more information on this consult \cite{pell}. 

Consider now a Fourier $n\ti n$ matrix. It is shown below that error-correcting pairs exist for codes generated by the rows of this Fourier matrix where the rows are taken in succession or in arithmetic sequence $k$ with $\gcd(n,k)=1$. 

Let $F= F_n$ be a Fourier
$n\ti n$ matrix with $\om$ as the element of order $n$. 

Denote the  rows of $F$ in order by $\{E_0, E_1, \ldots,
E_{n-1}\}$. It is easily checked that $E_iE_{n-i}\T=n$ and $E_iE_j\T = 0$
for $j\neq n-i \mod n$. Thus 

$$\begin{pmatrix}E_0 \\ E_1 \\ \vdots \\ E_{n-1}\end{pmatrix} (E_0\T,
E_{n-1}\T, E_{n-2}\T, \ldots, E_1\T) = nI_n$$

Call this the {\em Fourier Equation} for future reference. 

Note that if $H$ is a check matrix for a code then also $\al H$ is a
check matrix for the code for any $\al \neq 0$. 

We write out the details for the cases where the first $r$ rows are
taken as the generator matrix. The cases where rows are taken in
succession or where rows are taken in arithmetic sequence $k$ with $\gcd(n,k)=1$ are similar; in
all cases it requires getting error-correcting pairs and working from
there. 

The general Vandermonde case with restriction on cases where the
rows are taken in arithmetic sequence, is given in Section \ref{van1}.

Suppose then $\mathcal{C}$ is the code obtained by
taking the first $r$ rows of $F$. Thus $\mathcal{C}= \langle E_0,
E_1, \ldots, E_{r-1} \rangle$. Then $\mathcal{C}^\perp$ is $\langle
E_1, E_2, \ldots, E_{n-r}\rangle$ which can also be obtained by
eliminating the first $r$ columns of the second matrix on the left 
in the Fourier Equation.



Note that $E_i*E_j = E_{i+j}$ where suffices are taken $\mod n$. 
Let $A = \langle E_1, E_2 \ldots, E_{t+1}, \rangle, B = \langle
E_0, E_1, \ldots, E_{t-1}\rangle $ when $(n-r)$ is even and $t=\frac{n-r}{2}$, and 
let $A=\langle E_1, E_2, \ldots, E_{t+1} \rangle, B= \langle E_0, E_1, \ldots,
E_t \rangle $ when $(n-r)$ is odd and $t=\floor{\frac{n-r}{2}}$.

{\bf Then it may be verified that $A,B$ is a $t$-error correcting pair for
$\mathcal{C}$.}

Thus: 
 
\begin{enumerate}
\item $A*B \subseteq \mathcal{C}^\perp$
\item $k(A) > t$
\item $d(A) + d(\mathcal{C}) > n$
\item $d(B^\perp) > t$
\end{enumerate} 

This gives the following algorithm for locating and quantifying up to
$t$ errors for the code $\mathcal{C}$. 
In \cite{hurleycomp} the method of error-correcting pairs of Pellikaan
\cite{pell} is translated into an algorithm for decoding codes defined
by rows in succession or in (certain) arithmetic sequences of a
Vandermonde
/Fourier
matrix.   This may be applied directly here.

Let  $C$ be the $r \ti n$ generator matrix of $\mathcal{C}$. 
Suppose now $\al$ is a $1\ti r$ codeword,  that $\al C$ is sent but that $\al C + \ep$ is received  for a $1\ti n$ vector $\ep$ with at most $t$ non-zero entries. 
 
Assume $(n-r)$ is even; the other case is similar. 
Thus we are assuming $n-r=2t$. Now $\mathcal{C}^\perp$ is a check matrix for the code and thus $\ep E_1, \ep E_2, \ldots, \ep E_{n-r}$ are known by applying the check matrix to $\al C + \ep$. Let $\al_i=\ep E_i$ for $i=1,2,\ldots, n-r (=2t)$.

The algorithm then is:

\begin{Algorithm}\label{al1}

\begin{enumerate}

\item Find a non-zero solution of the kernel of the $t\ti (t+1)$ Hankel matrix

$\begin{pmatrix} \al_1 & \al_2 & \al_3 & \ldots, &\al_{t+1} \\ \al_2 & \al_3 & \al_4 & \ldots & \al_{t+2} \\ \vdots & \vdots & \vdots & \vdots & \vdots \\ \al_t & \al_{t+1} & \al_{t+2} & \ldots & \al_{2t} \end{pmatrix}$.

Call this solution $\underline{x}\T$ which is a $(t+1) \ti 1$ vector.

\item Let $\underline{a} = (E_1,E_2, \ldots, E_{t+1})\underline{x}\T$ which is a $1\ti n $ vector.





      (Any non-zero multiple of $\underline{a}$
      will suffice as we are only interested  in the zero entries of
      $\underline{a}$. Note that $\underline{a}$ is a $1\times n$ vector.)
\item Let $z(\underline{a})= \{j | 
a_j=0\}$  which is the set of locations of the zero coordinates of
$\underline{a}$. Suppose $z(\underline{a})=\{j_1,j_2, \ldots, j_t\}$ and denote this set by $J$.
\item Solve $s_J(x) = s(w)$. This reduces to solving the following.
Here $E_i=(E_{i,1}, E_{i,2}, \ldots, E_{i,n})$.
\begin{eqnarray}\label{est2}
\begin{pmatrix}E_{1,j_1} & E_{1,j_2} & \ldots &E_{1,j_t}\\ E_{2,j_1} &
 E_{2,j_2} & \ldots &E_{2,j_t} \\ \vdots & \vdots & \vdots &\vdots \\ 
E_{2t,j_1} & E_{2t,j_2} & \ldots
&E_{2t,j_t}\end{pmatrix}\begin{pmatrix}x_1 \\ x_2 \\ \vdots \\ x_t
			\end{pmatrix}
= \begin{pmatrix} \al_1 \\ \al_2 \\ \vdots \\ \al_{2t} \end{pmatrix}   
\end{eqnarray}

\item Now since in this case $E_{i,j} = \om^{i*j}$ the equation \ref{est2} may be put in the form 
\begin{eqnarray}\label{est3}
\begin{pmatrix} 1 & 1 & \ldots & 1\\ \om^{j_1} & \om^{j_2} & \ldots &\om^{j_t} \\ \om^{2j_1} &
 \om^{2j_2} & \ldots &\om^{2j_t} \\ \vdots & \vdots & \vdots &\vdots \\ 
\om^{(2t-1)j_1} & \om^{(2t-1)j_2} & \ldots
&\om^{(2t-1)j_t}\end{pmatrix}\begin{pmatrix}\om^{j_1}x_1 \\\om^{j_2} x_2 \\ \vdots \\  \om^{j_t}x_t
			\end{pmatrix}
= \begin{pmatrix} \al_1 \\ \al_2 \\ \vdots \\ \al_{2t} \end{pmatrix}   
\end{eqnarray}

(This form shows that the equation to be solved is a Vandermonde system containing roots of unity but not a (full) Fourier matrix.)
\item 
The value of $w$ is then the solution of  equations (\ref{est2}) or equivalently equations (\ref{est3})
 with entries in appropriate places as determined by $J$.
\end{enumerate}
 
\end{Algorithm}
\subsection{In arithmetic sequence}\label{general1} 

Suppose 
 $A$ is an 
$n\ti n$  Fourier  
 matrix with rows $\{E_0,E_1,\ldots, E_{n-1}\}$; these rows satisfy $E_i*E_j =
 E_{i+j}$. 

The  $E_jw$ are known for $j \in J= \{j_1, j_2,
\ldots, j_u\}$  where $u\geq 2t$. The elements in $J$ are in arithmetic
progression with difference $k$  satisfying $\gcd(n,k)=1$. 
Then $w$ is calculated by the following algorithm.
Let $\al_k = <w,F_{j_k}> = F_{j_k}w$ for $j_k\in J$.   
Define $F_i=E_{j_i}$ for $j_i \in J$ 
and $F_0=E_{j_1-k}$ with indices taken$\mod n$.
Let $F_i=(F_{i,1}, F_{i,2},\ldots, F_{i,n})$.  

\begin{Algorithm}\label{algor1}

\quad
\begin{enumerate}

\item Find a non-zero  element $x\T$ of the kernel of 
$E=\begin{pmatrix} \al_1 &\al_2 & \ldots & \al_{t+1} \\ \al_2 & \al_3 &
  \ldots & \al_{t+2} \\ \vdots & \vdots & \vdots & \vdots \\ \al_t &
  \al_{t+1} & \ldots & \al_{2t}\end{pmatrix}$.
\item Let $\underline{a}= (F_0,F_1,\ldots, F_{t})x\T$. (Any non-zero multiple of $a$
      will suffice as we are only interested  in the zero entries of
      $\underline{a}$. Note that $\underline{a}$ is a $1\times n$ vector.)
\item Let $z(\underline{a})= \{j | a_j=0\}$  which is the set of locations of the zero coordinates of
$\underline{a}$. Suppose $z(\underline{a})=\{j_1,i_2, \ldots, j_t\}$ and denote this set by $J$.
\item Solve $s_J(x) = s(w)$. This reduces to solving the following:
\begin{eqnarray}\label{est1}
\begin{pmatrix}F_{1,j_1} & F_{1,j_2} & \ldots &F_{1,j_t}\\ F_{2,j_1} &
 F_{2,j_2} & \ldots &F_{2,j_t} \\ \vdots & \vdots & \vdots &\vdots \\ 
F_{2t,j_1} & F_{2t,j_2} & \ldots
&F_{2t,j_t}\end{pmatrix}\begin{pmatrix}x_1 \\ x_2 \\ \vdots \\ x_t
			\end{pmatrix}
= \begin{pmatrix} \al_1 \\ \al_2 \\ \vdots \\ \al_{2t} \end{pmatrix}   
\end{eqnarray}

\item 
The value of $w$ is then the solution of these equations with entries in
appropriate places as determined by $J$.
\end{enumerate}
 
\end{Algorithm}

In Algorithm \ref{al1} it is shown that the equations (\ref{est2}) are  equivalent to a Vandermonde system of equations (\ref{est3}); similarly here  it can be  seen that the equations in (\ref{est1}) are equivalent to  a Vandermonde system with roots of unity as entries (but not the full Fourier matrix).  







\subsection{The general Vandermonde case}\label{vander}
 Working with a general Vandermonde matrix introduces difficulties as the inverse is not always nice to work with. However error-correcting algorithms can be formulated in many cases and we briefly discuss these cases here. \footnote{(This 
general Vandermonde case can be done  similar to that of Section 7  of \cite{hurleycomp} although in that paper the field is $\cc$.)  }


Consider the  Vandermonde matrix

$V=V(\all_1,\all_2,\ldots,\all_n) = \begin{pmatrix}
1&1&\ldots &1 \\ \all_1& \all_2& \ldots& \all_n \\ \vdots & \vdots &
\vdots & \vdots \\ \all_1^{n-1} & \all_2^{n-1} & \ldots &
\all_n^{n-1} \end{pmatrix}$ 

We assume the $\all_i$ are distinct and non-zero. 
 
Denote the rows of $V$ in order by $\{E_0,E_1,\ldots, E_{n-1}\}$. Then
$E_i*E_j=E_{i+j}$ as long as $i+j\leq n$. 
  
 Define
$E_k$ to be  $(\all_1^k,\all_2^k, \ldots, \all_n^k)$ for any $k\in
 \Z$. The rows of $V$ are $\{E_0,E_1,\ldots, E_{n-1}\}$ and these have been extended.  
\begin{lemma}\label{prod} $E_i*E_j = E_{i+j}$.
\end{lemma} 
\begin{proof} This is simply because $\be^i\be^j = \be^{i+j}$.
\end{proof}

Let $ \mathcal{C}^\perp =
\langle E_{j_1}, E_{j_2}, \ldots,
E_{j_u}\rangle$, where $u=2t$.  If $\mathcal{C}^\perp $ has rows in
arithmetic sequence with arithmetic difference $k$ 
and the ratios $\be_i\be_i^{-1}$ for $i\neq j$ in $V$ are not $k^{th}$ roots
of unity then $\mathcal{C}$ (the dual of $\mathcal{C}^\perp$) is an $(n,n-2t,2t+1)$ code, see Proposition \ref{vanderm},  and is $t$-error
correcting with  $C^\perp$ as the check matrix. 
Then also $\mathcal{C}$ has
an error correcting pair and a decoding   Algorithm may be derived. 
However it is not easy to describe $\mathcal{C}$
itself for this general Vandermonde case. 





Let $\al_i=<w,E_{j_i}> = E_{j_i}w\T$ for $j_i\in J$. Let $F_i=E_{j_i}$
for $j_i\in J$.
Thus $\al_i=<w,F_i>$.

\medskip 

\begin{Algorithm}\label{algor10} 

\quad
\begin{enumerate}[label=(\roman*),ref=(\roman*)]

\item Find a non-zero  element $v\T$ of the kernel of 
$E=\begin{pmatrix} \al_1 &\al_2 & \ldots & \al_{t+1} \\ \al_2 & \al_3 &
  \ldots & \al_{t+2} \\ \vdots & \vdots & \vdots & \vdots \\ \al_t &
  \al_{t+1} & \ldots & \al_{2t}\end{pmatrix}$.\label{eqs64}
\item Let $a= (F_1,F_2,\ldots, F_{t+1})v\T$. 

\item Let $z(a)= \{j | 
a_j=0\}$  which is the set of locations of the zero coordinates of
$a$. 
Suppose $z(a)=\{i_1,i_2, \ldots, i_t\}$ and denote this set by $J$.\label{eqs30}
\item\label{matrix10} Solve $s_J(x) = s(w)$. This reduces to solving the following:
\begin{eqnarray}\label{est10} 
\begin{pmatrix}\all_{i_1}^{j_1} & \all_{i_2}^{j_1} & \ldots
  &\all_{i_t}^{j_1} \\ \all_{i_1}^{j_2} &
 \all_{i_2}^{j_2} & \ldots &\all_{i_t}^{j_2} \\ \vdots & \vdots & \vdots &\vdots \\ 
\all_{i_1}^{j_{2t}} & \all_{i_2}^{j_{2t}} & \ldots
&\all_{i_t}^{j_{2t}} \end{pmatrix}\begin{pmatrix}x_1 \\ x_2 \\ \vdots \\ x_t
			\end{pmatrix}
= \begin{pmatrix} \al_1 \\ \al_2 \\ \vdots \\ \al_{2t} \end{pmatrix}   
\end{eqnarray}

Since the entries in the matrix of (\ref{est10}) have arithmetic difference $k$ giving  that $j_s=
i_1+(s-1)k$ for $1\leq s \leq 2t$, the equation (\ref{est10}) is equivalent to 
\begin{eqnarray}\label{est11}
\begin{pmatrix} 1 & 1 & \ldots
  & 1 \\ \all_{i_1}^{k} &
 \all_{i_2}^{k} & \ldots &\all_{i_t}^{k} \\ \vdots & \vdots & \vdots &\vdots \\ 
\all_{i_1}^{(2t-1)k} & \all_{i_2}^{(2t-1)k} & \ldots
&\all_{i_t}^{(2t-1)k} \end{pmatrix}\begin{pmatrix}\be_{i_1}^{j_1}x_1
  \\ \be_{i_2}^{j_1}x_2 \\ \vdots  \\ \be_{i_t}^{j_1}x_t
			\end{pmatrix}
= \begin{pmatrix} \al_1 \\ \al_2 \\ \vdots \\ \al_{2t} \end{pmatrix}   
\end{eqnarray}
\item 
Then $x=(x_1,x_2,\ldots,x_t)$ is obtained from these equations
(\ref{est11}) (or from (\ref{est10})) and $w$ has   
 entries $x_i$ in positions  as determined by $J$ and zeros elsewhere.
\end{enumerate}
 
\end{Algorithm}

The matrix in (\ref{est11}) is a Vandermonde matrix. It
is sufficient to solve the first $t$ equations and the inverse of such a
$t\times t$ Vandermonde type matrix may be obtained in $O(t^2)$
operations.  In connection with item \ref{eqs64}, finding a non-zero
element of the kernel of a Hankel $t\times (t+1)$ matrix can be done in $O(t^2)$
or less operations.  


\section{Code to a rate and error capability}\label{rate} Suppose an mds code of rate $R=\frac{r}{n}$ is
required. 

It is required to obtain over a finite field a Fourier $n\ti n$ matrix.

We can take $n$ to be as large as necessary as $\frac{r}{n}
= \frac{rs}{sn}$ for any positive integer $s$.

Let $p$ be a prime not dividing $n$. Then by Euler's theorem, 
$p^{\phi(n)}\equiv 1 \mod n$ where $\phi$ is the Euler $\phi$ function. 
Thus $p^{\phi(n)}- 1 = nq$ for some positive integer $q$. Consider the
field $\F=GF(p^{\phi(n)})$. Then a primitive generator, $\be$ say, of the field has
order  $(p^{\phi(n)}-1)=nq$. Then $\be^q = \om$ has order $n$ in
$\F$. Construct the Fourier $n\ti n$ matrix, $F_{n}$, over $\F$ using $\om$ as the
element of order $n$.   Now by the method of the previous sections,
$(n,r,n-r+1)$ codes can be constructed with efficient decoding
algorithms from $F_n$. 

For a prime $p$ not dividing $n$ we know that there exists a positive
integer $q$ such that $p^q \equiv 1 \mod n$. So for best results take $q$
to be the smallest such positive integer and do the calculations in
$GF(p^q)$. 

If $n$ is odd then $2\not\vert n$ and so the Fourier matrix can be
obtained over $GF(2^k)$ for some $k$, where $2^k \equiv 1 \mod n$. For
example if $n=103$ then the order of $2$ mod $103$ is $51$ and so the
Fourier matrix may be obtained over $GF(2^{51})$. Making the calculations
over $GF(2^s)$ has advantages in that codes over such a field may be
transmitted as binary digits. 

Suppose a rate $R=\frac{r}{n}$ is required and in addition $t$ errors
may need to be corrected. 
Then it is required that $t= \floor{\frac{n-r}{2}}$. Assume $n-r$ is
even; the other case is similar. Then it is required that 
$t=\frac{n-r}{2} =\frac{n(1-R)}{2}$.   
 
\subsection{Examples}\label{example1} 

\subsubsection{Rate $\frac{5}{7}$} Suppose a rate of $\frac{5}{7}$ is required and that $t=50$
errors should be correctable. This gives that $\frac{n(1-\frac{5}{7})}{2}
=t=50$ which requires $n=350$. Thus a code $(350,250,101)$ is
required. Thus construct a Fourier $350\ti 350$ matrix over a field. Now
$3$ is a prime not dividing $n=350$ and the order of $3 \mod 350$ is
$60$. Thus this required Fourier matrix exists over $GF(3^{60})$. Also the
order of $11 \mod 350$ is $15$ and the field $GF(11^{15})$ may also be
used.  A little investigation shows that the order of $43 \mod 350$ is $4$ so the field $GF(43^4)$ could be used.

Let the required rate again be $\frac{5}{7}$ and now it is required that $t=49$ errors be correctable. This gives that $\frac{n(1-\frac{5}{7})}{2}
=t=49$ which requires $n=343$. Require a $(343, 245, 99)$ code. 
  Since $n$ is odd it is possible to find a field $GF(2^s)$ which has a $343^{th}$ root of unity. The order of $2 \mod 343$ is $147$ so the field $GF(2^{147})$ could be used but  would be large. However the order of $19 \mod 343$ is $6$ so it is possible to work in $GF(19^6)$. 

Let the required rate again be $\frac{5}{7}$ and now it is required that $t=48$ errors be correctable. This gives that $\frac{n(1-\frac{5}{7})}{2}
=t=48$ which requires $n=336$. Require a $(336, 240, 97)$ code. 
  Now note that $337$ is prime so can work in the prime field $GF(337)$ which involves modular arithmetic. An element of order $336$ is required in $GF(337)$ and this is easily found. For example  the order of $10 \mod 337$ is $336$ and thus $\om = 10 \mod 337$ may be used as the element of order $336$ in forming the Fourier $336\ti 336$ matrix over $GF(337)$. Here the arithmetic is  modular arithmetic,  which is nice. 
 
\subsubsection{Rate  $\frac{31}{32}$}
Suppose a rate $\frac{31}{32}$ is specified and we would like the code to correct at least 50 errors. Then for $(n,r,n-r+1)$ we need $n-r \geq 2* 50 = 100$ and so need for $R=\frac{31}{32}$ that $n*\frac{1}{32} \geq 100$ which is $n\geq 3200$. We would also like to work with modular arithmetic. 
 Now notice that $3201$ is not a prime but that $3203$ is a prime. Thus let $n=3202$ and construct the code $(3202, 3102,101)$ over the prime field $GF(3203)$. This code has rate slightly less ($0.0000019..$) than $\frac{31}{32}$. To have rate of $\frac{31}{32}$  and still work over a prime field take $n=104*32=3328$ and then $n+1=3329$ is prime. Here we work over the prime field $GF(3329)$ and get the code $(3328,3224,105)$ which can correct  $52$ errors. 

The order of $2 \mod 3203$ is $3220$ so $2 \mod 3203$ may be used in as
the element of order $3202$ for the Fourier $3202\ti 3202$ matrix in
$GF(3203)$. All the non-zero elements of $GF(3203)$ are used for this
Fourier matrix. From it codes of all forms $(3203, r, 3203-r+1)$ may be
obtained for $1\leq r \leq 3202 $.      

\subsection{Rate $\frac{3}{4}$; correct lots}
Suppose for example a code is required that could correct $50$ errors
and have a rate of $\frac{3}{4}$. The code of smallest length satisfying
these conditions is one of the form $(400,300,101)$, How could such a code be constructed? One way is to
construct a Fourier $400 \ti 400$ matrix and select three quarters of the 
rows in order so that an mds
code is generated. Thus select 300 rows in
sequence from the Fourier matrix. What is the smallest field over which
such a $400\ti 400$ Fourier can exist? What is the field of smallest
characteristic over which such a Fourier matrix can exist? Now
$\phi(400)= 160$ so we need the smallest field $GF(p^s)$ such that $p^s
\equiv 1 \mod 400$ with $\gcd(400,p)=1$ and $s|160$ as necessary requirements. Here it is
found that  $7^4 \equiv 1 \mod 400$ so we can use the
field $GF(7^4)$. This is the smallest field for which there exists a
$400\ti 400$ Fourier matrix. Now $7^4=2401$ and thus the field is
relatively small and its characteristic is small. The $400\ti 400$ Fourier matrix over $GF(7^4)$ can be used to find the $(400,300,101)$ code  but it can also be used to find  $(400,r, 401-r) $ codes over $GF(7^4)$. For example $(400,350, 51)$ code can correct $25$ errors and $(400,200,201)$ code over $GF(7^4)$ can correct $100$ errors.

Consider constructing a code of rate $\geq \frac{3}{4}$ and which can correct $50$ errors  but now require the code to be over $GF(p)$ for a prime $p$. 
Now the order of the non-zero elements of $GF(p)$ is $p-1$ and we require $n=p-1
\geq 400$. It turns out that $p=401$ is a prime which is the least prime
$p$ for which $p\geq 400$. Now the order of $2
\mod 401$ is $200$ so using $2 \in GF(401)$ doesn't work but the order
of $3 \mod 401$ is $400$. Hence let  $\om = 3 \mod 401$ and form the
$400 \ti 400 $ Fourier matrix $F_{400}$ over $GF(401)$ with $\om$ as a primitive
$400^{th}$ root of unity. Now choose the first $300$ rows of $F_{400}$
or any consecutive $300$ rows in $F_{400}$ gives a $(400,300,101)$ code
as required.  

The calculations are done $\mod 401$. Error correcting pairs are
also obtainable  from the unit Fourier scheme which are then used for
the efficient decoding algorithms. 

Which are better, the codes over $GF(7^4)$ or the codes over $GF(401)$? 

\subsection{Remark}

Many such constructions are possible. Codes over $GF(2^s)$, for $s $ not
too large, and codes over prime fields may be particularly useful. 

\subsection{`Optimal' codes from a given  field}\label{best} Suppose the field $GF(p^s)$ is given and it is required to construct the best possible codes with coefficients from this field. Let $n=p^s-1$. Then there exists an element $\om$ of order $n$ in $GF(p^s)$ and every non-zero element is a power of this  generator. Form the Fourier $n\ti n$ matrix using $\om$ as a primitive $n^{th}$ root of unity. Unit-derived codes are then formed using rows of $F$ in succession or else in arithmetic sequence $k$ satisfying $\gcd(n,k) = 1$. For any $ 1\leq r \leq n$, mds codes of the form $(n,r,n-r+1)$ may be constructed from the rows of this Fourier matrix. The Fourier matrix uses all the non-zero elements of $GF(p^s)$. 

These are the best performing codes from $GF(p^s)$; the lengths are $p^{s}-1$ and all possible rates $\frac{r}{n}$ with $r\leq n$ are available.

\section{Shannon}\label{Shannon}
Here we relate the previous Hamming results to Shannon results.


For a given rate $1\geq R>0$ the previous sections give methods for
constructing $(n,r,n-r+1)$ codes where $\frac{r}{n}=R$. The probability
of error is the probability that more than $k=\floor{\frac{n-r}{2}}$
errors occur in the binomial distribution with $p$ the probability that
an error occurs at a component. Here $\mu = np$. 

Chernoff's bounds \cite{cher} give the following:

$\Pr[X\geq (1+\de)\mu] \leq (\frac{e^\de}{(1+\de)^{1+\de})})^{\mu} \leq e^{\frac{-\de^2}{2+\de}\mu} = e^{\frac{-\de^2}{2+\de}np}$  
 for
$\de > 0$. 


$\Pr[X\leq (1-\de)\mu] \leq (\frac{e^{-\de}}{(1-\de)^{1-\de})})^{\mu} < 
(\frac{e^{-\de}}{e^{-\de+\de^2/2}})^\mu < e^{-\de^2\mu/2}$ 
for $0< \de \leq 1$. 


Now consider a code $(n,r, n-r+1)$ which can correct $k =
\floor{\frac{n-r}{2}}$ errors and has an efficient decoding algorithm. Assume $n-r$ is even; the other case
is similar; thus $k=\frac{n-r}{2}$. Now $r = nR$ where $R$ is the
rate. 

For the first Chernoof inequality to hold it is required that   $(1+\de)np= k+1 = 1+ \frac{n-r}{2} = 1+ \frac{n(1-R)}{2}$.
Thus $(1 + \de) = \frac{1}{np} + \frac{1-R}{2p}$ and thus 
$\de= \frac{1-R}{2p}-1 + \frac{1}{np}$. We require $\de >0$ and so
require $\frac{1-R}{2p}- 1 + \frac{1}{np} >0 $. Multiply across by $2p$
and this requires $(1-R) + 2/n > 2p$ which is equivalent to $R< 1-2p +
2/n$. For $n$ large enough make $R<1-2p + 2/n$. Then 
the probability of error is $< e^{\frac{-\de^2}{2+\de}np}$. 


Now $R < 1- 2p+ \frac{2}{n}$ means that $R$ can be as close to $1-2p$ as necessary and then the probability of error is less than $ e^{-\ga n}$ for some $\ga > 0$. Note that $p< \frac{1}{2}$ implies that $1-2p > 0$ and then $R> 0$ also for $n$ big enough.

For the second Chernoff inequality to hold requires $(1-\de)\mu = \frac{n-r}{2}$
which is $(1-\de)np = \frac{n(1-R)}{2}$; this requires
$(1-\de) \frac{1-R}{2p}$ and hence $ -\de = \frac{1-R}{2p}-1$. Now $\de
> 0 $ requires $\frac{1-R}{2p} -1 < 0$ in which case require
$R>1-2p$. For $\de \leq 1$ requires $-\de \geq -1$ in which case it is
required that $\frac{1-R}{2p}- 1 \geq -1$ from which it is required that
$\frac{1-r}{2p} \geq 0$ from which it is required that $1\geq R$, which
is true. Thus the second Chernoff inequality can be applied for $R>
1-2p$. Thus for $R> 1-2p$ the probability of no error is less than
$e^{-\de^2\mu/2} $.  Thus for $n$ big enough 
the probability of error is $\geq \frac{1}{2}$.

In order to construct a $(n,r,n-r+1)$ code over a finite field by the
 unit-derived method with Fourier/Vandermonde matrices 
it is necessary to have a field $F=GF(p^k)$ such  that $n |
 (p^k-1)$. The 
rate is $R=\frac{r}{n}$ 
 and $n$ can be taken to be as large as necessary as $\frac{r}{n}
= \frac{rs}{sn}$ for any positive integer $s$.

Let $p$ be a prime not dividing $n$. Then by Euler's theorem, 
$p^{\phi(n)}\equiv 1 \mod n$ where $\phi$ is the Euler $\phi$ function. 
Thus $p^{\phi(n)}- 1 = nq$ for some positive integer $q$. Consider the
field $\F=GF(p^{\phi(n)})$. Then a primitive generator $\be$ of the field has
order  $p^{\phi(n)}-1$. Then $\be^q = \om$ has order $n$ in
$F$. Construct the Fourier $n\ti n$ matrix, $F_{n}$,  over $\F$ using $\om$ as the
element of order $n$.   Now by the method of the previous sections,
$(n,r,n-r+1)$ codes can be constructed with efficient decoding
algorithms from $F_n$. 

For a prime $p$ not dividing $n$ we know that there exists a positive
integer $q$ such that $p^q \equiv 1 \mod n$. So for best results take $q$
to be the smallest such positive integer and do the calculations in
$GF(p^q)$. 

If a `rate' $H$ is required which is not a rational number then take the
`nearest' rational number to $H$.

\section{Complexity}\label{complexity} 
The decoding calculations require finding a non-zero element in the kernel of a Hankel $t\ti (t+1)$ matrix. Finding the kernel of 
an $t\ti (t+1)$ Hankel matrix can be done in $O(t^2)$ operations. 
Super-fast algorithms of $O(t\log^2t)$ have been proposed with which to find the kernel of a Hankel $t\ti (t+1)$ matrix.  

It is then required to solve a system of $2t \ti t$ equations where the coefficients on the left of the matrix are roots of unity; solving the first $t\ti t$ equations is sufficient. The matrix of the system of $t\ti t$ equations reduces to a  Vandermonde matrix whose entries are roots of unity. 
Now the system can be solved in $O(t^2)$ operations. The entries of the Vandermonde matrix are roots of unity in a finite field which make the calculations easier and stable. 

Consider the case where the encoder is the first part (first rows) of a Fourier matrix.
Thus we are in the situation $\begin{pmatrix} A \\ B \end{pmatrix} (C,D) = I$
where $\begin{pmatrix} A \\ B\end{pmatrix}$ is a Fourier matrix and $(C,D) $ is a multiple ($\frac{1}{n}$ for length $n$) of a Fourier matrix.
 
The encoding is $\al \mapsto \al A$ where $A$ is part of a Fourier matrix $F=\begin{pmatrix} A  \\ B \end{pmatrix}$. Thus by adding $0^s$ to the end of $\al$ to get $\bar{\al}$ of length $n$ ensures the encoding can be done by (Fast) Fourier Transform if necessary. 

Similarly the decoding can be done by (Fast) Fourier Transform when the the errors have been eliminated as $\al A C = \al$ and $C$ is part of a Fourier matrix  $(C,D)$. In fact $\al A (C, D) = (\al AC, \al AD) = (\al, 0)$. 


Thus the calculations can all be done in at worst the maximum of $O(n\log n)$ and  $O(t^2)$ for length $n$ and error-correction $t$. The $t^2$ is a function of the error-correction capability $t$. Now in the vast majority of  cases  the required distance $2t+1$ satisfies $t\leq \sqrt(n)$; in these  cases all the calculations can be done in at worst $O(n\log n)$ calculations. If super fast calculations of the kernel of a Hankel $t\ti (t+1)$ matrix are employed as proposed then the calculations can be done in $\max\{O(n\log n), O(t\log^2t)\}$ operations. This is certainly of $O(n\log n)$ when $\log^2t \leq n$ or $\log t \leq \sqrt{\log n}$.


\begin{thebibliography}{99}

\bibitem{blahut} Richard E.\ Blahut, {\em Algebraic Codes for data 
  transmission}, Cambridge University Press, 2003.
\bibitem{cher} Herman Chernoff, ``A measure of Asymptotic Efficiency for Tests of a Hypothesis Based on a sum of Observations'', Annals of Math. Stats., 23, No.4 493-507, 1952. See also the many lecture notes on the topic available on-line and elsewhere. 
\bibitem{hur1} Paul Hurley and Ted Hurley, ``Codes from zero-divisors
  and units in group rings'', Int. J. Inform. and Coding Theory, 1, 
  57-87, 2009.
\bibitem{hur2} Paul Hurley and Ted Hurley, ``Block codes from matrix
  and group rings'', Chapter 5, 159-194, in {\em Selected Topics in
    Information and Coding Theory}, eds. I. Woungang, S. Misra,
  S.C. Misma, World Scientific 2010.
\bibitem{hur3} Paul Hurley and Ted Hurley, ``LDPC and convolutional
	codes from matrix and group rings'',  Chapter 6, 195-237, in
	{\em Selected Topics in 
    Information and Coding Theory},  eds. I. Woungang, S. Misra,
  S.C. Misma, World Scientific 2010.

\bibitem{hurconv5} Ted Hurley, ``Convolutional codes from units in
	matrix and group rings'', Inter. J. of Pure and Applied
	Mathematics 50(3), 431-463, 2009. 
\bibitem{hurconv2} Ted Hurley, ``Convolutional codes from unit
	schemes'', ArXiv 1412.1695, 22 pp., 2016. 
\bibitem{hur7} Ted Hurley, ``Group rings and rings of matrices'',
  Inter. J. Pure \& Appl. Math., 31, no.3, 2006, 319-335.
\bibitem{hurleycomp} Ted Hurley, ``Solving underdetermined systems with error correcting codes'', Intl. J. Information and Coding Theory, Vol 4, no. 4, 201-221, 2017.
\bibitem{hur100} Ted Hurley, ``Cryptographic schemes, key exchange, public key'', Intl. J. of Pure and Applied Maths., 93, 6,897-927, 2014. 

\bibitem{pell} R. Pellikaan, ``On decoding by error location and
  dependent sets of error positions'', Discrete Math., Vol.\ 106/107,
  369-381, 1992. 



\bibitem{duur} I.  Duursma \& R. K\"otter, `` Error-locating pairs for cyclic codes''. IEEE Trans. Inform. Theory,  40, 1108–1121,  1994.
\bibitem{hur0} Paul Hurley and Ted Hurley, ``Module codes in group
rings'', ISIT2007, Nice, 2007, 1981-1985. 
\bibitem{hurley} Barry Hurley and Ted Hurley, ``Systems of MDS codes
	from units and idempotents'', Discrete Math., 335, 81-91, 2014. 
\bibitem{hurley44} Ted Hurley, ``Self-dual, dual-containing and related
  quantum codes from group rings'', arXiv:0711.3983.
\bibitem{hurley33} Ted Hurley, Paul McEvoy and Jakub Wenus, ``Algebraic
	constructions of LDPC codes with no short cycles'', 
Intl.\ J.\ of Inform.\ and Coding Theory, Vol 1, Issue 3, 285-297, 2010. 

\bibitem{gap} `GAP -- Groups, Algorithms and Programming',
  www.gap-system.org
\bibitem{mceliececrypt} R. J. McEliece, ``A Public-Key Cryptosystem Based On Algebraic Coding Theory'', DSN (Deep Space Network) Progress Report 42-44: 114–116, 1978.  
\bibitem{mceliece} R.J. McEliece, {\em Theory of Information and
 Coding}, 2nd ed., Cambridge University Press, 2002.


\end{thebibliography}
\end{document}